%% file: main.tex
\documentclass[runningheads]{llncs}
\usepackage[T1]{fontenc}
\usepackage{graphicx}
\usepackage{amsmath}
\usepackage{amssymb}

\usepackage{tcolorbox}
\usepackage{tikz}

\usepackage[colorlinks=true,citecolor=blue]{hyperref}

\usetikzlibrary{decorations.pathreplacing,calligraphy}

\usepackage{xcolor}

\definecolor{mysoftred}{RGB}{255,117,101}
\definecolor{mysoftyellow}{RGB}{255,255,101}
\definecolor{mysoftgreen}{RGB}{117,255,101}

\newcommand{\sd}{{\sc Signed Domination}}

\newcommand{\ds}{{\sc Dominating Set}}

\newcommand{\mrss}{{\sc Multidimensional Relaxed Subset Sum}}

\begin{document}
\title{On the Complexity of Signed Domination\thanks{Extended abstract of this paper has appeared in IWOCA 2026~\cite{reddy}.}}
%
%
\author{Sangam Balchandar Reddy\orcidID{0000-0002-5848-3821}}
\authorrunning{S. B. Reddy}
%
\institute{School of Computer and Information Sciences, \\ University of Hyderabad, Telangana, India \vspace{3mm} \\
\email{21mcpc14@uohyd.ac.in}}
\maketitle              
\begin{abstract}
Given a graph $G = (V, E)$, a signed dominating function is a function $f: V \rightarrow \{-1, 1\}$ such that for every vertex $u \in V$, $\sum\limits_{v \in N[u]} f(v) \geq 1$. The weight of $f$ is defined as $\sum\limits_{u \in V} f(u)$. The objective of the \sd{} problem is to compute a signed dominating function $f$ of minimum weight. The problem is known to be NP-complete even when restricted to bipartite, chordal, and planar graphs. In this paper, we extend the known complexity results for the \sd{} problem. Since the problem is NP-complete on chordal graphs, we study its complexity on split graphs, a subclass of chordal graphs, and show that it remains NP-complete. Moreover, as the problem is W[2]-hard parameterized by weight, we investigate its parameterized complexity with respect to structural parameters. We prove that the problem is W[1]-hard when parameterized by feedback vertex set number (and hence by treewidth and clique-width). Motivated by this hardness result, we consider more restrictive parameters, neighbourhood diversity and twin cover number, and present FPT algorithms.
\end{abstract}
 \input{1_introduction}
 \input{2_preliminaries}
 \input{3_split}
 \input{4_fvs}
 \input{5_nd}
 \input{6_tc}
 \input{7_conclusion}
 \bibliographystyle{plain}
\bibliography{references}
\end{document}

%% file: 1_introduction.tex
\section{Introduction}

The \ds{} (DS) problem is one of the earliest problems shown to be NP-complete by \cite{karp}. Given a graph $G = (V, E)$, a set $S \subseteq V$ is a dominating set if for each vertex $u \in V\setminus S$, $|N(u) \cap S| \geq 1$. The goal of the DS problem is to compute a dominating set of minimum size.

Since its inception, numerous variants of the DS problem have been introduced. One such variant is the \sd{} (SD) problem. The concept of signed domination was introduced by \cite{dunbar1995signed}. Signed domination generalizes classical domination by allowing vertices to contribute either positively or negatively to the domination condition, thereby providing greater modeling flexibility. A signed dominating function on graph $G$ is a function $f: V \rightarrow{} \{-1, 1\}$ such that for each vertex $u \in V$, $\sum\limits_{v \in N[u]}f(v) \geq 1$. The weight of a signed dominating function $f$ is $\sum\limits_{u \in V} f(u)$. The minimum weight over all signed dominating functions for $G$ is denoted by $\gamma_{s}(G)$. The objective of the SD problem is to compute $\gamma_{s}(G)$.

\medskip
 We formally define the decision version of the SD problem as follows.
\begin{tcolorbox}
{
\sd{}: \newline
\textit{Input:} An instance $I$ = $(G, k)$, where $G=(V,E)$ is an undirected graph and an integer $k$.\newline
\textit{Output:} YES, if there exists a function $f: V \rightarrow{} \{-1,1\}$ such that $\sum\limits_{u \in V} f(u) \leq k$ and $\sum\limits_{v \in N[u]} f(v) \geq 1$ for each vertex $u \in V$; NO otherwise.
}
\end{tcolorbox}
\cite{henning1995algorithmic} proved that the SD problem is NP-complete on bipartite, and chordal graphs. They also showed that the problem admits a linear-time algorithm on trees. \cite{damaschke2001minus} then showed that the problem remains NP-complete on planar graphs with maximum degree three. Further, \cite{lee2008variations} proved that the problem is NP-complete even on doubly chordal graphs, planar bipartite graphs and chordal bipartite graphs. Later, \cite{zheng2013kernelization} showed that the problem is NP-complete on subcubic grid graphs. They also presented a kernel of size $\frac{k^2+k}{2}$ for general graphs, where $k$ is the number of vertices with label $1$. Recently, \cite{lin2015algorithms} proved that the problem is W[2]-hard parameterized by weight. They also presented constant-factor approximation algorithms for subcubic graphs, graphs of maximum degree four and graphs of maximum degree five. In addition, they proved that the problem is FPT on subcubic graphs.

\medskip
Several bounds have been established for the SD problem on various graph classes. In particular, \cite{dunbar1995signed} derived both upper and lower bounds on $\gamma_{s}(G)$. For $r$-regular graphs, \cite{henning1996inequalities} showed that $\gamma_{s}(G) \geq \frac{n}{r+1}$ when $r$ is even, and $\gamma_{s}(G) \geq \frac{2n}{r+1}$ when $r$ is odd. \cite{favaron1996signed} subsequently provided upper bounds, proving that $\gamma_{s}(G) \leq n \cdot \frac{(r+1)^2}{r^2+4r-1}$ for odd $r$, and $\gamma_{s}(G) \leq n \cdot \frac{r+1}{r+3}$ for even $r$. For general graphs, \cite{zhang1999note} established the lower bound $\gamma_{s}(G) \geq 2\left\lceil \frac{-1+\sqrt{1+8n}}{2} \right\rceil - n$.

\medskip
Signed domination has applications in systems where elements exert both positive and negative influence, making it useful in real-world network analysis. For instance, in sensor networks, nodes can be in active (+1) or inactive ($-$1) states, and the goal is to ensure that every node receives the positive support from its neighbours. This framework helps in designing efficient and fault-tolerant systems by minimizing resource usage while maintaining overall network stability. Similar interpretations arise in social and biological networks, where entities may promote or inhibit activity, and signed domination provides a mathematical tool to analyze and optimize such balanced interactions. 
\medskip 

\noindent 
\textbf{Our Results.}
 In this paper, we investigate the SD problem in the realm of classical and parameterized complexity, and obtain the following results. 
 \begin{itemize}
    \item We analyze the complexity of the problem on split graphs (a subclass of chordal graphs) and prove that it remains NP-complete.
     \item We then show that the problem is W[1]-hard parameterized by feedback vertex set number. From this result, it follows that the problem is W[1]-hard parameterized by treewidth or clique-width.
     \item Since the problem is W[1]-hard parameterized by clique-width, we further investigate its parameterized complexity for larger parameters neighbourhood diversity and twin cover number, and present FPT algorithms. The relationship between various graph parameters can be seen in Fig. \ref{fig: Fig1a}.
 \end{itemize}
 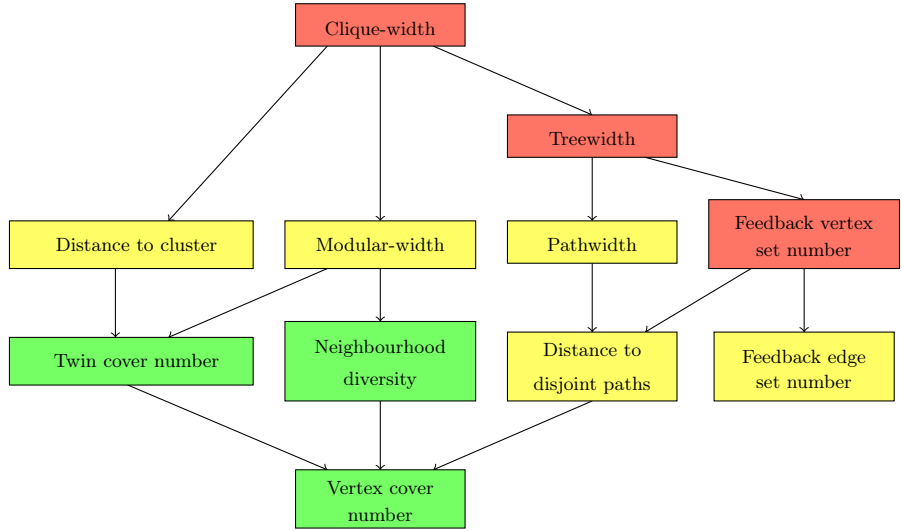
\begin{figure} [t]
\centering
    \begin{tikzpicture} [thick,scale=1.4, every node/.style={scale=0.8}]
        
        \draw[thin, <-] (-1.5, 0) -- (0, 0.65); 
        \draw[thin, <-] (0, 1.3) -- (0, 1.95); 
        \draw[thin, <-] (0, 2.35) -- (0, 2.95); 
        \draw[thin, <-] (0, 3.35) -- (-1.5, 4); 

        \draw[thin, <-] (2, 1.3) -- (2, 1.9); 
        \draw[thin, <-] (2, 2.55) -- (0.5, 2.95); 
        \draw[thin, <-] (-2, 0) -- (-2, 0.65);
         
        \draw[thin, <-] (-2.5, 0) -- (-4.5, 0.85);
        \draw[thin, ->] (-4.5, 1.9) -- (-4.5, 1.25);

        \draw[thin, <-] (-2, 1.4) -- (-2, 1.9); 
        \draw[thin, <-] (-2, 2.35) -- (-2, 4); 
        \draw[thin, <-] (-4, 2.35) -- (-2.5, 4); 
        \draw[thin, ->] (-2.5, 1.9) -- (-4, 1.25); 
        \draw[thin, ->] (1.5, 1.9) -- (0.5, 1.3);

        \draw[thin, fill=mysoftgreen] (-2.8, -0.55) rectangle (-1.2, 0);
        \draw[thin, fill=mysoftyellow] (-0.8, 1.3) rectangle (0.8, 0.65);
        \draw[thin, fill=mysoftyellow] (-0.8, 1.95) rectangle (0.8, 2.35);
        \draw[thin, fill=mysoftred] (-0.8, 2.95) rectangle (0.8, 3.35);
        
        \draw[thin, fill=mysoftred] (-2.8, 4) rectangle (-1.2, 4.4);
        \draw[thin, fill=mysoftgreen] (-2.9, 1.4) rectangle (-1.1, 0.65);
        \draw[thin, fill=mysoftyellow] (-2.9, 1.9) rectangle (-1.1, 2.35);
        \draw[thin, fill=mysoftyellow] (1.15, 1.3) rectangle (2.85, 0.65);
        \draw[thin, fill=mysoftred] (1.1, 1.9) rectangle (2.9, 2.55);

        \draw[thin, fill=mysoftgreen] (-5.5, 1.25) rectangle (-3.2, 0.8);
        \draw[thin, fill=mysoftyellow] (-5.5, 1.9) rectangle (-3.2, 2.35);

        \draw (-2, 4) circle (0cm) node[anchor=south]{Clique-width};
        \draw (0, 3) circle (0cm) node[anchor=south]{Treewidth};
        \draw (0, 2) circle (0cm) node[anchor=south]{Pathwidth};
        \draw (0, 1) circle (0cm) node[anchor=south]{Distance to};
        \draw (0, 0.65) circle (0cm) node[anchor=south]{disjoint paths};
        \draw (-2, -0.3) circle (0cm) node[anchor=south]{Vertex cover};
        \draw (-2, -0.55) circle (0cm) node[anchor=south]{number};

        \draw (-2, 2) circle (0cm) node[anchor=south]{Modular-width};
        \draw (-2, 1) circle (0cm) node[anchor=south]{Neighbourhood};
        \draw (-2, 0.7) circle (0cm) node[anchor=south]{diversity};

        \draw (-4.3, 2) circle (0cm) node[anchor=south]{Distance to cluster};
        \draw (-4.3, 0.9) circle (0cm) node[anchor=south]{Twin cover number};
        
        \draw (2, 2.2) circle (0cm) node[anchor=south]{Feedback vertex};
        \draw (2, 1.95) circle (0cm) node[anchor=south]{set number};
        \draw (2, 0.9) circle (0cm) node[anchor=south]{Feedback edge};
        \draw (2, 0.7) circle (0cm) node[anchor=south]{set number};

    \end{tikzpicture}
    \caption{Hasse diagram representing the relation between various structural parameters. A directed edge from a parameter $k_1$ to $k_2$ indicates that $k_1 \leq f(k_2)$ for some computable function $f$. Parameters marked by green are proved to be FPT. Parameters marked by red are proved to be W[1]-hard and the complexity of the parameter marked by yellow is unknown.}
    \label{fig: Fig1a}
\end{figure}

%% file: 2_preliminaries.tex
 \section{Preliminaries}
 \noindent Consider a graph $G = (V, E)$ with $V$ as the vertex set and $E$ as the edge set. For a vertex $u \in V$, the open neighbourhood of $u$ is denoted by $N(u) = \{v: (u,v) \in E\}$ and the closed neighbourhood of $u$ is denoted by $N[u] = N(u) \cup \{u\}$. The open neighbourhood of a set $T \subseteq V$ is denoted by $N(T) = (\bigcup\limits_{u \in T}^{} N(u)) \setminus T$ and the closed neighbourhood by $N[T] = N(T) \cup T$. The degree of a vertex $u$ is represented by $d(u)$ and $d(u) = |N(u)|$. A graph $G$ is $r$-regular if $d(u) = r$ for every vertex $u \in V$. A cubic graph is a 3-regular graph. A vertex of degree one is called a pendant vertex. A path of length $l$ refers to a path with $l$ edges and $l+1$ vertices.

Let $f$ be a signed dominating function. For a vertex $u \in V$, the value of $f(u)$ is also called the \textbf{label} of $u$. For a vertex $u$, \textbf{labelSum} is defined as the sum of the labels of all the vertices in its closed neighbourhood, i.e., $\sum\limits_{v \in N[u]} f(v)$. For a set $T$, \textbf{weight} of $T$ is the sum of the labels of all the vertices in the set, i.e., $\sum\limits_{u \in T} f(u)$. 
Apart from this, standard notations as defined in \cite{WEST} are used.
\medskip

\noindent \textit{Parameterized Complexity.} A problem is said to be fixed-parameter tractable (FPT) with respect to a parameter $k$, if it can be solved by an algorithm with running time $\mathcal{O}(f(k) \cdot n^{\mathcal{O}(1)})$, where $f$ is a computable function and $n$ is the size of the input. The notation $\mathcal{O}^*(f(k))$ is used to suppress polynomial factors in $n$ and to denote running times of this form. A problem is W[$\star$]-hard with respect to a parameter if it is believed not to be fixed-parameter tractable. In parameterized complexity, there exists a hierarchy of complexity classes: FPT $\subseteq$ W[1] $\subseteq$ W[2] $\subseteq$ ... $\subseteq$ XP. For more information on \textit{parameterized complexity}, the reader is referred to \cite{MC}.
\medskip

\noindent We now define the class of graphs used in this paper. 
\medskip

\noindent \emph{Split graphs.}
    A graph $G = (V, E)$ is a split graph if the vertex set $V$ can be partitioned into two sets $V_1$ and $V_2$ such that the subgraph induced by $V_1$ is a clique and the subgraph induced by $V_2$ is an independent set in $G$. 
\medskip

\noindent Next, we define the graph parameters used in this paper.
\medskip

\noindent \emph{Feedback vertex set number.}
     For a graph $G = (V, E)$, the feedback vertex set number is the cardinality of a smallest vertex subset $F \subseteq V$ such that $G \setminus F$ is a disjoint union of trees.
\medskip

\noindent \emph{Neighbourhood diversity.}
     Let $G = (V, E)$ be a graph. Two vertices $u$ and $v$ are of the same type if and only if $N(u) \setminus \{v\} = N(v) \setminus \{u\}$. The neighbourhood diversity of $G$ is the smallest value of $t$ for which there exists a partition of $V$ into $t$ sets $V_1, V_2, ..., V_t$ such that all the vertices in each set are of the same type.
\medskip

\noindent \emph{Twin cover number.} For a graph $G = (V, E)$, an edge $(u, v) \in E$ is a twin edge of $G$ if $N_G[u] = N_G[v]$. The parameter \textit{twin cover number} is the cardinality of the smallest set $T \subseteq V$ such that every edge in $G$ is either a twin edge or incident to a vertex in $T$.
\medskip

 \indent A key result used in Section~\ref{sectionhard} of this paper is given as follows.
 \begin{lemma} \label{obs: O1} 
    Let $f$ be a signed dominating function. If $u \in V$ with $d(u) = 1$ and $v \in N(u)$, then $f(u) = f(v) = 1$. 
    \end{lemma}
\begin{proof}
    Consider a vertex $u \in V$ with $d(u)$ = 1. Let $v$ be the only neighbour of $u$. To ensure that $u$ has a positive labelSum, neither $u$ nor $v$ can be assigned the label $-1$. If either vertex were labeled $-1$, the labelSum of $u$ would be at most zero. Therefore, we must have $f(u) = f(v) = 1$. \qed
\end{proof}

%% file: 3_split.tex
\section{NP-complete on split graphs}
In this section, we prove that the SD problem is NP-complete on split graphs, a subclass of chordal graphs. We provide a polynomial-time reduction from the DS problem on cubic graphs to prove that the SD problem is NP-complete.

\medskip
The complexity result related to the DS problem on cubic graphs is given as follows.
\begin{theorem}[\cite{Kikuno}]\label{kikunothm}
    DS problem is NP-complete even when restricted to cubic graphs.
\end{theorem}
Given an instance $I = (G, k)$ of the DS problem where $G$ is a cubic graph, we construct an instance $I' = (G', 2k)$ of the SD problem as follows.
\begin{itemize}
    \item Two copies of the vertex set \(V(G)\), denoted by \(A\) and $X$, are included in \(G'\). A vertex $u_i \in V(G)$ is denoted by $a_i$ in $A$ and $x_i$ in $X$. Thus, $A = \{a_1,a_2, ..., a_{n}\}$ and $X =\{x_1, x_2, ..., x_n\}$.
    \item Two additional sets $B$ of size $4n$ and $Y$ of size $2n$ are created. Let $B = \{b_1,b_2, ..., b_{4n}\}$ and $Y = \{y_1,y_2, ..., y_{2n}\}$.
    \item For each pair of vertices \( a_i \in A \) and \( x_j \in X \) with $i, j \in [n]$, an edge is added between \( a_i \) and \( x_j \) if and only if \( u_i \in N_G(u_j) \). In addition, an edge is added between $a_i$ and $x_i$ for each $i \in [n]$.
    \item For each $i \in [n]$, vertex $x_i \in X$ is made adjacent to four vertices $b_{4i-3}$,$b_{4i-2}$,$b_{4i-1}$ and $b_{4i}$.
    \item For each $i \in [2n]$, vertex $y_i \in Y$ is made adjacent to both $b_{2i-1}$ and $b_{2i}$.
    \item Finally, a clique is formed on the vertex set $A \cup B$.
\end{itemize}
\begin{figure} [h!]
    \centering
        \begin{tikzpicture} [scale = 0.45]
    \filldraw[black] (-2, 0) circle (3pt);
    \filldraw[black] (0, 1) circle (3pt);
    \filldraw[black] (0, -1) circle (3pt);
    \filldraw[black](2, 1) circle (3pt);
    \filldraw[black] (2, -1) circle (3pt);
    \filldraw[black] (4, 0) circle (3pt);
    \draw[thin] (-1.9, 0.1) -- (-0.1, 0.9);
    \draw[thin] (-1.9, -0.1) -- (-0.1, -0.9);
    \draw[thin] (-1.85, 0) -- (3.85, 0);
    \draw[thin] (0, 0.85) -- (0, -0.85);
    \draw[thin] (2, 0.85) -- (2, -0.85);
    \draw[thin] (0.15, 1) -- (1.85, 1);
    \draw[thin] (0.15, -1) -- (1.85, -1);
    \draw[thin] (2.1, 0.9) -- (3.9, 0.1);
    \draw[thin] (2.1, -0.9) -- (3.9, -0.1);
    \filldraw (-2.55, -0.35) circle (0cm) node[anchor=south]{$u_1$};
    \filldraw (0, 1.1) circle (0cm) node[anchor=south]{$u_2$};
    \filldraw (0, -2) circle (0cm) node[anchor=south]{$u_3$};
    \filldraw (2, 1.1) circle (0cm) node[anchor=south]{$u_4$};
    \filldraw (2, -2) circle (0cm) node[anchor=south]{$u_5$};
    \filldraw (4.75, -0.35) circle (0cm) node[anchor=south]{$u_6$};
    \filldraw (4, -5) circle (0cm) node[anchor=south]{};
    \filldraw (1.5, -6.5) circle (0cm) node[anchor=south]{};
    \filldraw (1.5, -4.2) circle (0cm) node[anchor=south]{$G$};
    \end{tikzpicture}
     \begin{tikzpicture} [thick,scale=1.25, every node/.style={scale=0.95}]
     \draw (0, -0.5) ellipse (1 and 3.15);
        \draw (4,-1.1) ellipse (0.8 and 2.5);
        \draw[thin, gray] (-0.99, 0.4) -- (0.99, 0.4);
        \draw[thin, gray] (3.2, -0.6) -- (4.8, -0.6);

        \draw[blue] (4, 0.9) -- (0, 0);
        \draw[blue] (4, 0.9) -- (0, -0.3);
        \draw[blue] (4, 0.9) -- (0, -0.6);
        \draw[blue] (4, 0.9) -- (0, -0.9);

        \draw[thin, gray] (4, -0.1) -- (0, -2.2);
        \draw[thin, gray] (4, -0.1) -- (0, -2.5);
        \draw[thin, gray] (4, -0.1) -- (0, -2.8);
        \draw[thin, gray] (4, -0.1) -- (0, -3.1);


        


        \draw[blue] (4, 0.9) -- (0, 2.3);
        \draw[thin, gray] (4, 0.7) -- (0, 2.3);
        \draw[thin, gray] (4, 0.5) -- (0, 2.3);
        \draw[thin, gray] (4, -0.1) -- (0, 2.3);
        \draw[thin, gray] (4, 0.7) -- (0, 2);
        \draw[blue] (4, 0.9) -- (0, 2);
        \draw[thin, gray] (4, 0.5) -- (0, 2);
        \draw[thin, gray] (4, 0.3) -- (0, 2);
        \draw[thin, gray] (4, 0.5) -- (0, 1.7);
        \draw[blue] (4, 0.9) -- (0, 1.7);
        \draw[thin, gray] (4, 0.7) -- (0, 1.7);
        \draw[thin, gray] (4, 0.3) -- (0, 1.7);
        \draw[thin, gray] (4, 0.3) -- (0, 1.4);
        \draw[thin, gray] (4, 0.7) -- (0, 1.4);
        \draw[thin, gray] (4, 0.1) -- (0, 1.4);
        \draw[thin, gray] (4, -0.1) -- (0, 1.4);
        \draw[thin, gray] (4, 0.1) -- (0, 1.1);
        \draw[thin, gray] (4, 0.5) -- (0, 1.1);
        \draw[thin, gray] (4, 0.3) -- (0, 1.1);
        \draw[thin, gray] (4, -0.1) -- (0, 1.1);
        \draw[thin, gray] (4, -0.1) -- (0, 0.8);
        \draw[blue] (4, 0.9) -- (0, 0.8);
        \draw[thin, gray] (4, -0.1) -- (0, 0.8); 
        \draw[thin, gray] (4, 0.1) -- (0, 0.8);    

        \draw[thin, gray] (0, 0) -- (4, -1.15);       
        \draw[thin, gray] (0, -0.3) -- (4, -1.15);        
        \draw[thin, gray] (0, -0.6) -- (4, -1.35);    
        \draw[thin, gray] (0, -0.9) -- (4, -1.35);
        
        \draw[thin, gray] (0, -2.2) -- (4, -2.95);
        \draw[thin, gray] (0, -2.5) -- (4, -2.95);   
        \draw[thin, gray] (0, -2.8) -- (4, -3.15);
        \draw[thin, gray] (0, -3.1) -- (4, -3.15);

        \filldraw[gray] (0, -1.5) circle (0.5pt) node[anchor=south]{};
        \filldraw[gray] (0, -1.6) circle (0.5pt) node[anchor=south]{};
        \filldraw[gray] (0, -1.7) circle (0.5pt) node[anchor=south]{};
        
        \filldraw[gray] (4, -2.2) circle (0.5pt) node[anchor=south]{};
        \filldraw[gray] (4, -2.3) circle (0.5pt) node[anchor=south]{};
        \filldraw[gray] (4, -2.1) circle (0.5pt) node[anchor=south]{};

        \filldraw (0, 1.7) circle (1pt) node[anchor=south]{};
        \filldraw (0, 1.1) circle (1pt) node[anchor=south]{};
        \filldraw (0, 1.4) circle (1pt) node[anchor=south]{};
        \filldraw (0, 0.8) circle (1pt) node[anchor=south]{};
        \filldraw (0, 2) circle (1pt) node[anchor=south]{};
        \filldraw (0, 2.3) circle (1pt) node[anchor=south]{};

        \filldraw (0, -2.2) circle (1pt) node[anchor=south]{};
        \filldraw (0, -2.5) circle (1pt) node[anchor=south]{};
        \filldraw (0, -2.8) circle (1pt) node[anchor=south]{};
        \filldraw (0, -3.1) circle (1pt) node[anchor=south]{};

        \filldraw (-0.25, 1.55) circle (0cm) node[anchor=south]{$a_3$};
        \filldraw (-0.25, 1.25) circle (0cm) node[anchor=south]{$a_4$};
        \filldraw (-0.25, 0.95) circle (0cm) node[anchor=south]{$a_5$};
        \filldraw (-0.25, 0.65) circle (0cm) node[anchor=south]{$a_6$}; 
        \filldraw (-0.25, 1.85) circle (0cm) node[anchor=south]{$a_2$};
        \filldraw (-0.25, 2.15) circle (0cm) node[anchor=south]{$a_1$};  

        \filldraw (4.25, 0.75) circle (0cm) node[anchor=south]{$x_1$};
        \filldraw (4.25, 0.55) circle (0cm) node[anchor=south]{$x_2$};
        \filldraw (4.25, 0.35) circle (0cm) node[anchor=south]{$x_3$};
        \filldraw (4.25, 0.15) circle (0cm) node[anchor=south]{$x_4$};
        \filldraw (4.25, -0.05) circle (0cm) node[anchor=south]{$x_5$};
        \filldraw (4.25, -0.25) circle (0cm) node[anchor=south]{$x_6$};

        \filldraw (-0.25, -0.1) circle (0cm) node[anchor=south]{$b_1$};
        \filldraw (-0.25, -0.4) circle (0cm) node[anchor=south]{$b_2$};
        \filldraw (-0.25, -0.7) circle (0cm) node[anchor=south]{$b_3$};
        \filldraw (-0.25, -1) circle (0cm) node[anchor=south]{$b_4$};    


        \filldraw (-0.25, -2.3) circle (0cm) node[anchor=south]{$b_{21}$};
        \filldraw (-0.25, -2.6) circle (0cm) node[anchor=south]{$b_{22}$};
        \filldraw (-0.25, -2.9) circle (0cm) node[anchor=south]{$b_{23}$};
        \filldraw (-0.25, -3.2) circle (0cm) node[anchor=south]{$b_{24}$};
        
        \filldraw (4.25, -1.25) circle (0cm) node[anchor=south]{$y_1$};
        \filldraw (4.25, -1.45) circle (0cm) node[anchor=south]{$y_2$};
        
        \filldraw (4.25, -3.1) circle (0cm) node[anchor=south]{$y_{11}$};
        \filldraw (4.25, -3.3) circle (0cm) node[anchor=south]{$y_{12}$};
        
        \filldraw (0, 0) circle (1pt) node[anchor=south]{};
        \filldraw (0, -0.3) circle (1pt) node[anchor=south]{};
        \filldraw (0, -0.6) circle (1pt) node[anchor=south]{};
        \filldraw (0, -0.9) circle (1pt) node[anchor=south]{};

        \filldraw (4, 0.9) circle (1pt) node[anchor=south]{};
        \filldraw (4, 0.7) circle (1pt) node[anchor=south]{};
        \filldraw (4, 0.5) circle (1pt) node[anchor=south]{};
        \filldraw (4, 0.3) circle (1pt) node[anchor=south]{};
        \filldraw (4, 0.1) circle (1pt) node[anchor=south]{};
        \filldraw (4, -0.1) circle (1pt) node[anchor=south]{};
        
        \filldraw (4, -1.15) circle (1pt) node[anchor=south]{};
        \filldraw (4, -1.35) circle (1pt) node[anchor=south]{};

        \filldraw (4, -2.95) circle (1pt) node[anchor=south]{};
        \filldraw (4, -3.15) circle (1pt) node[anchor=south]{};
        
        \filldraw (-1.2, 1) circle (0cm) node[anchor=south]{$A$};
        \filldraw (-1.3, -1.6) circle (0cm) node[anchor=south]{$B$};
        \filldraw (0.1, -4.1) circle (0cm) node[anchor=south]{clique};
        \filldraw (4.2, -4.1) circle (0cm) node[anchor=south]{independent set};
        \filldraw (5.05, 0.2) circle (0cm) node[anchor=south]{$X$};
        \filldraw (5.1, -2.1) circle (0cm) node[anchor=south]{$Y$};
        \filldraw (2, -4.7) circle (0cm) node[anchor=south]{$G'$};
     \end{tikzpicture}
    \caption{Reduced instance of the SD problem constructed from an instance of the DS problem. Here, $k = 2$. Edges between the vertices of the clique $A \cup B$ are not shown in the figure. The edges incident on the vertex $x_1$ are colored in blue.}
    \label{fig:fig2}
\end{figure}
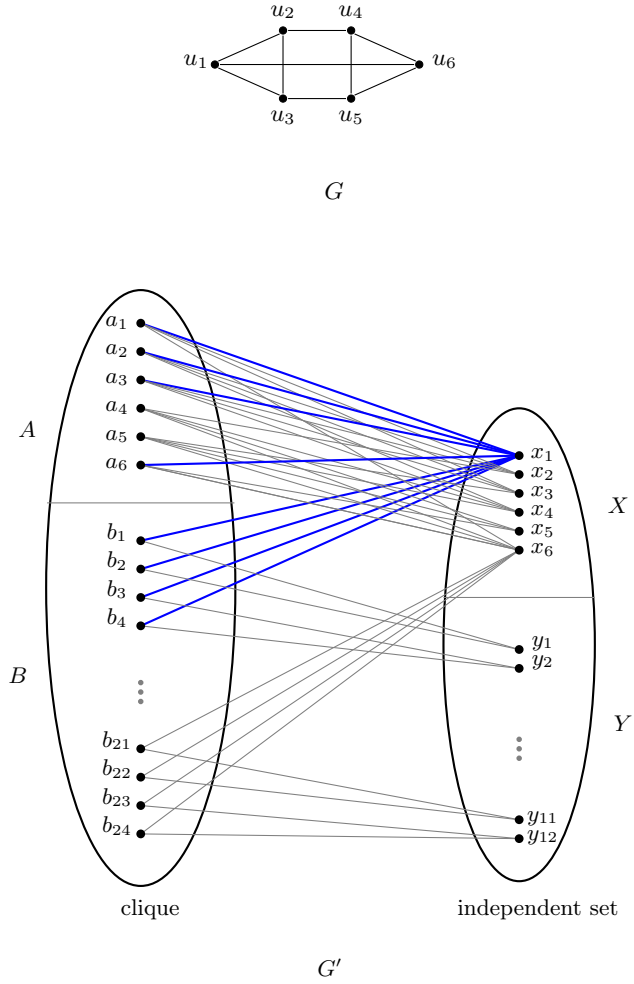
This concludes the construction of $G'$. See Fig. \ref{fig:fig2} for an illustration.

\medskip
Observe that $G'$ is a split graph, where the subgraph induced by $A \cup B$ forms a clique and the subgraph induced by $X \cup Y$ forms an independent set.
\begin{lemma}~\label{lemma1}
    $G$ has a dominating set of size at most $k$ if and only if $G'$ has a  signed dominating function of weight at most $2k$.
\end{lemma}
\begin{proof}
$[\Rightarrow]$ Let $S$ be a dominating set of size at most $k$. The corresponding signed dominating function $f$ is given as follows.
\[
f(w) = 
\begin{cases}
-1, & \text{if } w \in \bigcup\limits_{u_i \notin S} \{a_i\} \cup X \cup Y, \\
1, & \text{if } w \in \bigcup\limits_{u_i \in S} \{a_i\} \cup B
\end{cases}
\]
\begin{itemize}   
    \item Since $G$ is cubic, each vertex $u \in X$ is adjacent to eight vertices in $A \cup B$. Vertex $u$ has a labelSum of at least one as all four of its neighbours in $B$ have label $1$ and at least one of its neighbours in $A$ has label $1$.
    \item Each vertex in $Y$ has a labelSum of one due to the positive labels for both of its neighbours in $B$.
    \item The weight of the clique $A \cup B$ is $3n+2k$ which is at least $3n+2$ (as $k \geq$ 1). Therefore, each vertex in $A$ has a labelSum of at least $3n-2$, which is positive for all $n \geq 1$.
    \item Similarly, each vertex in $B$ has a labelSum of at least $3n$, which is at least three for all $n \geq 1$.
\end{itemize}
The weight of $A$ is at most $2k-n$, that of $B$ is $4n$, that of $X$ is $-n$ and that of $Y$ is $-2n$. Since each vertex in $G'$ has a positive labelSum, we conclude that $f$ is a signed dominating function of weight at most $2k$.
\medskip

\noindent $[\Leftarrow]$ Let $f$ be a signed dominating function of weight at most $2k$. The corresponding dominating set $S$ is given as follows. 
\begin{itemize}
    \item For each $i \in [2n]$, vertex $y_i$ is adjacent to two vertices in $B$, so its closed neighbourhood has size three. To ensure a positive labelSum for $y_{i}$, at most one of $b_{2i-1}$, $b_{2i}$ and $y_i$ can be labeled $-1$. Since $y_i$ belongs to the independent set $X \cup Y$, we label $y_i$ with $-1$ and both $b_{2i-1}$ and $b_{2i}$ with $1$. Hence, the weight of $B$ is $4n$ and the weight of $Y$ is $-2n$.
    \item Since the weight of $B \cup Y$ is $2n$, the weight available for $A \cup X$ is at most $2k-2n$. Therefore, at most $k$ vertices in $A \cup X$ can be labeled $1$, while the remaining vertices must be labeled $-1$.
    \item At this point, every vertex $u \in X$ has a labelSum of four from its neighbours in $B$. For $u$ to have a positive labelSum, either $u$ itself or at least one of its neighbours in $A$ must be labeled 1. Since at most $k$ vertices in $A \cup X$ can receive label $1$, it follows that we must assign label 1 to $k$ vertices in $A$ such that every vertex in $X$ is adjacent to at least one vertex in $A$ labeled 1.
    \item The $k$ vertices in $A$ that are labeled 1 will correspond to a $k$-vertex dominating set $S$ for $G$.
\end{itemize} \qed
\end{proof}

Hence, from Theorem \ref{kikunothm} and Lemma \ref{lemma1}, we arrive at the following theorem.
\begin{theorem}
    SD problem is NP-complete on split graphs.
\end{theorem}

%% file: 4_fvs.tex
\section{W[1]-hard parameterized by feedback vertex set number}\label{sectionhard}
In this section, we study the parameterized complexity of the SD problem parameterized by feedback vertex set number.
We provide a parameterized reduction from a well known W[1]-hard problem, \mrss{} (MRSS). The MRSS problem is defined as follows.
\begin{tcolorbox}
{
\textbf{Problem.} \mrss{} \newline
\textbf{Input.} An integer $k$, a set $S = \{s_1, s_2, ..., s_n\}$ of vectors with $s_i \in \mathbb{N}^k$ for every $i$ with $1 \leq i \leq n$, a target vector $t \in \mathbb{N}^k$ and an integer $m$. \newline
\textbf{Parameter.} $k+m$\newline
\textbf{Output.} Does there exists a subset $\mathcal{S} \subseteq S$ with $|\mathcal{S}| \leq m$ such that $\sum\limits_{s \in \mathcal{S}} s \geq t$?
}
\end{tcolorbox}
The parameterized complexity of the MRSS problem has been established in the literature as follows.
\begin{theorem} [\cite{GNN}]~\label{mrss}
    MRSS problem is W[1]-hard parameterized by $k+m$, even when all integers in the input are in unary.
\end{theorem}

 \medskip
\noindent \textbf{Construction.} Consider an instance $I = (k, m, S, t)$ of the MRSS problem. Let $\max(s)$ denote the value of the largest coordinate of a vector $s$. We construct an instance $I' = (G, k')$ of the SD problem in the following way. 
\begin{itemize}
    \item $k$ vertices $\{u_1, u_2 ,..., u_k\}$ are created in $G$.
    \item A new set of vectors $S'$ is constructed from $S$ as follows: 
        \subitem For each vector $s \in S$, each coordinate of $s$ is increased by 1 and the resulting vector is included in $S'$.
        \subitem $\max(t)$ copies of the vector $(2,2)$ are added to $S'$.
    \item The vector $t'$ is obtained from $t$ by adding $m$ to each of its coordinates.
    \item For each $j \in [k]$, a set $P_j$ is introduced with $(\sum\limits_{s_i \in S'} s_i(j)) - 2t'(j)$ vertices.
    \item For each $j \in [k]$, vertex $u_j$ is made adjacent to every vertex in the set $P_j$.
    \item For each $s_i \in S'$, $\max(s_i)$ paths of length two are created. For $l \in [\max(s_i)]$, vertices of these paths are denoted by $a_i^l, b_i^l$ and $c_i^l$, with $b_i^l$ adjacent to both $a_i^l$ and $c_i^l$. A vertex $d_i$ is then created and made adjacent to every vertex in $\bigcup\limits_{l \in [\max(s_i)]}c_i^l$.
    \item Let $A_{s_i}$ = $\bigcup\limits_{l \in [\max(s_i)]}$ $a_i^l$, $B_{s_i}$ = $\bigcup\limits_{l \in [\max(s_i)]}$ $b_i^l$ and $C_{s_i}$ = $\bigcup\limits_{l \in [\max(s_i)]}$ $c_i^l$.
    \item For each $s_i \in S'$ and $j \in [k]$, vertex $u_j$ is made adjacent to exactly $s_i(j)$ vertices in $A_{s_i}$ arbitrarily.
\end{itemize}
This concludes the construction of the reduced instance $G$. See  Fig.~\ref{fig: Fig2} for an illustration. 

\medskip
The purpose of adding $\max(t)$ copies of the vector $(2,2)$ to $S'$ is to guarantee that, for each $j \in [k]$, the value of $(\sum\limits_{s_i \in S'} s_i(j)) - 2t'(j)$ is positive. This ensures that the set $P_j$, for each $j \in [k]$, contains at least one vertex.

\medskip
We set $k' = \sum\limits_{j  \in [k]} (\sum\limits_{s_i \in S'} s_i(j) - 2t'(j))$+ $k +\sum\limits_{s_i \in S'}$ ($\max(s_i)-1$) + $2m$ and the set corresponding to feedback vertex set number is $\bigcup\limits_{j \in [k]} \{u_j\}$ of size $k$.
\begin{figure} [t]
\centering
    \begin{tikzpicture} [thick,scale=0.4, every node/.style={scale=0.85}]
        
        \filldraw (3.5, 2) circle (0.15cm);
        \filldraw (5, 2) circle (0.15cm);
        \filldraw (6.5, 2) circle (0.15cm);
        \filldraw (11, 2) circle (0.15cm);
        \filldraw (12.5, 2) circle (0.15cm);
        \filldraw (14, 2) circle (0.15cm);
        \filldraw (18, 2) circle (0.15cm);
        \filldraw (20, 2) circle (0.15cm);
        \filldraw (22, 2) circle (0.15cm);
        \filldraw (24, 2) circle (0.15cm);
        \filldraw (26, 2) circle (0.15cm);
        \filldraw (28, 2) circle (0.15cm);
        \filldraw (30, 2) circle (0.15cm);
        \filldraw (32, 2) circle (0.15cm);

        \filldraw (3.5, -1) circle (0.15cm);
        \filldraw (5, -1) circle (0.15cm);
        \filldraw (6.5, -1) circle (0.15cm);
        \filldraw (11, -1) circle (0.15cm);
        \filldraw (12.5, -1) circle (0.15cm);
        \filldraw (14, -1) circle (0.15cm);
        \filldraw (18, -1) circle (0.15cm);
        \filldraw (20, -1) circle (0.15cm);
        \filldraw (22, -1) circle (0.15cm);
        \filldraw (24, -1) circle (0.15cm);
        \filldraw (26, -1) circle (0.15cm);
        \filldraw (28, -1) circle (0.15cm);
        \filldraw (30, -1) circle (0.15cm);
        \filldraw (32, -1) circle (0.15cm);

        \filldraw (3.5, -4) circle (0.15cm);
        \filldraw (5, -4) circle (0.15cm);
        \filldraw (6.5, -4) circle (0.15cm);
        \filldraw (11, -4) circle (0.15cm);
        \filldraw (12.5, -4) circle (0.15cm);
        \filldraw (14, -4) circle (0.15cm);
        \filldraw (18, -4) circle (0.15cm);
        \filldraw (20, -4) circle (0.15cm);
        \filldraw (22, -4) circle (0.15cm);
        \filldraw (24, -4) circle (0.15cm);
        \filldraw (26, -4) circle (0.15cm);
        \filldraw (28, -4) circle (0.15cm);
        \filldraw (30, -4) circle (0.15cm);
        \filldraw (32, -4) circle (0.15cm);

        \filldraw (5, -7) circle (0.15cm);
        \filldraw (12.5, -7) circle (0.15cm);
        \filldraw (19, -7) circle (0.15cm);
        \filldraw (23, -7) circle (0.15cm);
        \filldraw (27, -7) circle (0.15cm);
        \filldraw (31, -7) circle (0.15cm);

        \filldraw (5, 9) circle (0.15cm);
        \filldraw (31, 9) circle (0.15cm);

        \draw[thin] (3.5, 2) -- (5, 9);
        \draw[thin] (5, 2) -- (5, 9);  
        \draw[thin] (11, 2) -- (5, 9); 
        \draw[thin] (12.5, 2) -- (5, 9); 
        \draw[thin] (14, 2) -- (5, 9); 
        \draw[thin] (18, 2) -- (5, 9);
        \draw[thin] (20, 2) -- (5, 9);
        \draw[thin] (22, 2) -- (5, 9);
        \draw[thin] (24, 2) -- (5, 9); 
        \draw[thin] (26, 2) -- (5, 9);
        \draw[thin] (28, 2) -- (5, 9);
        \draw[thin] (30, 2) -- (5, 9);
        \draw[thin] (32, 2) -- (5, 9);

        \draw[thin] (3.5, 2) -- (31, 9); 
        \draw[thin] (5, 2) -- (31, 9);  
        \draw[thin] (6.5, 2) -- (31, 9); 
        \draw[thin] (11, 2) -- (31, 9);
        \draw[thin] (12.5, 2) -- (31, 9);
        \draw[thin] (18, 2) -- (31, 9);
        \draw[thin] (20, 2) -- (31, 9);
        \draw[thin] (22, 2) -- (31, 9);
        \draw[thin] (24, 2) -- (31, 9);
        \draw[thin] (26, 2) -- (31, 9);
        \draw[thin] (28, 2) -- (31, 9);
        \draw[thin] (30, 2) -- (31, 9);
        \draw[thin] (32, 2) -- (31, 9);

        \draw[thin] (3.5, 2.0) -- (3.5, -1); 
        \draw[thin] (5, 2.0) -- (5, -1); 
        \draw[thin] (6.5, 2.0) -- (6.5, -1); 
        \draw[thin] (11.0, 2.0) -- (11, -1); 
        \draw[thin] (12.5, 2.0) -- (12.5, -1); 
        \draw[thin] (14.0, 2.0) -- (14, -1); 
        \draw[thin] (18, 2.0) -- (18, -1);
        \draw[thin] (20, 2.0) -- (20, -1);
        \draw[thin] (22, 2.0) -- (22, -1);
        \draw[thin] (26, 2.0) -- (26, -1);
        \draw[thin] (30, 2.0) -- (30, -1);
        \draw[thin] (24, 2.0) -- (24, -1);
        \draw[thin] (28, 2.0) -- (28, -1);
        \draw[thin] (32, 2.0) -- (32, -1);

        \draw[thin] (3.5, -4) -- (3.5, -1);  
        \draw[thin] (5, -4) -- (5, -1);
        \draw[thin] (6.5, -4) -- (6.5, -1); 
        \draw[thin] (11.0, -4) -- (11, -1); 
        \draw[thin] (12.5, -4) -- (12.5, -1); 
        \draw[thin] (14.0, -4) -- (14, -1); 
        \draw[thin] (18, -4) -- (18, -1);
        \draw[thin] (20, -4) -- (20, -1);
        \draw[thin] (22, -4) -- (22, -1);
        \draw[thin] (26, -4) -- (26, -1);
        \draw[thin] (30, -4) -- (30, -1);
        \draw[thin] (24, -4) -- (24, -1);
        \draw[thin] (28, -4) -- (28, -1);
        \draw[thin] (32, -4) -- (32, -1);

        \draw[thin] (3.5, -4) -- (5, -7); 
        \draw[thin] (6.5, -4) -- (5, -7); 
        \draw[thin] (5, -4) -- (5, -7); 
        \draw[thin] (11.0, -4) -- (12.5, -7); 
        \draw[thin] (14.0, -4) -- (12.5, -7); 
        \draw[thin] (12.5, -4) -- (12.5, -7); 
        \draw[thin] (18, -4) -- (19, -7);
        \draw[thin] (20, -4) -- (19, -7);
        \draw[thin] (22, -4) -- (23, -7);
        \draw[thin] (24, -4) -- (23, -7);
        \draw[thin] (26, -4) -- (27, -7);
        \draw[thin] (28, -4) -- (27, -7);
        \draw[thin] (30, -4) -- (31, -7);
        \draw[thin] (32, -4) -- (31, -7);
        
        \filldraw (5, 12) circle (0.15cm);
        \filldraw (3.5, 12) circle (0.15cm);
        \filldraw (6.5, 12) circle (0.15cm);

        \draw[thin] (5, 11.9) -- (5, 9.1);
        \draw[thin] (3.57, 11.93) -- (4.93, 9.07);
        \draw[thin] (6.43, 11.93) -- (5.07, 9.07);
        
        \filldraw (31, 12) circle (0.15cm);
        \filldraw (29.5, 12) circle (0.15cm);
        \filldraw (32.5, 12) circle (0.15cm);

        \draw[thin] (31, 11.9) -- (31, 9.1);
        \draw[thin] (29.57, 11.93) -- (30.93, 9.07);
        \draw[thin] (32.43, 11.93) -- (31.07, 9.07);

        \draw (2.8, 1.25) circle (0cm) node[anchor=south]{$a_1^1$};        
        \draw (4.4, 1.25) circle (0cm) node[anchor=south]{$a_1^2$};       
        \draw (7.3, 1.25) circle (0cm) node[anchor=south]{$a_1^3$};     
        \draw (10.3, 1.25) circle (0cm) node[anchor=south]{$a_2^1$};
        \draw (11.9, 1.25) circle (0cm) node[anchor=south]{$a_2^2$};
        \draw (14.8, 1.25) circle (0cm) node[anchor=south]{$a_2^3$};
        \draw (17.3, 1.25) circle (0cm) node[anchor=south]{$a_3^1$};
        \draw (19.4, 1.25) circle (0cm) node[anchor=south]{$a_3^2$};
        \draw (21.3, 1.25) circle (0cm) node[anchor=south]{$a_4^1$};
        \draw (23.4, 1.25) circle (0cm) node[anchor=south]{$a_4^2$};
        \draw (25.3, 1.25) circle (0cm) node[anchor=south]{$a_5^1$};
        \draw (27.4, 1.25) circle (0cm) node[anchor=south]{$a_5^2$};
        \draw (29.3, 1.25) circle (0cm) node[anchor=south]{$a_6^1$};
        \draw (32.7, 1.25) circle (0cm) node[anchor=south]{$a_6^2$};

        \draw (2.8, -1.25) circle (0cm) node[anchor=south]{$b_1^1$};        
        \draw (4.4, -1.25) circle (0cm) node[anchor=south]{$b_1^2$};      
        \draw (7.3, -1.25) circle (0cm) node[anchor=south]{$b_1^3$};     
        \draw (10.3, -1.25) circle (0cm) node[anchor=south]{$b_2^1$};
        \draw (11.9, -1.25) circle (0cm) node[anchor=south]{$b_2^2$};
        \draw (14.8, -1.25) circle (0cm) node[anchor=south]{$b_2^3$};
        \draw (17.3, -1.25) circle (0cm) node[anchor=south]{$b_3^1$};
        \draw (19.4, -1.25) circle (0cm) node[anchor=south]{$b_3^2$};
        \draw (21.3, -1.25) circle (0cm) node[anchor=south]{$b_4^1$};
        \draw (23.4, -1.25) circle (0cm) node[anchor=south]{$b_4^2$};
        \draw (25.3, -1.25) circle (0cm) node[anchor=south]{$b_5^1$};
        \draw (27.4, -1.25) circle (0cm) node[anchor=south]{$b_5^2$};
        \draw (29.3, -1.25) circle (0cm) node[anchor=south]{$b_6^1$};
        \draw (32.7, -1.25) circle (0cm) node[anchor=south]{$b_6^2$};

        \draw (2.8, -4.25) circle (0cm) node[anchor=south]{$c_1^1$};       
        \draw (4.4, -4.25) circle (0cm) node[anchor=south]{$c_1^2$};       
        \draw (7.3, -4.25) circle (0cm) node[anchor=south]{$c_1^3$};     
        \draw (10.3, -4.25) circle (0cm) node[anchor=south]{$c_2^1$};
        \draw (11.9, -4.25) circle (0cm) node[anchor=south]{$c_2^2$};
        \draw (14.8, -4.25) circle (0cm) node[anchor=south]{$c_2^3$};
        \draw (17.3, -4.25) circle (0cm) node[anchor=south]{$c_3^1$};
        \draw (19.4, -4.25) circle (0cm) node[anchor=south]{$c_3^2$};
        \draw (21.3, -4.25) circle (0cm) node[anchor=south]{$c_4^1$};
        \draw (23.4, -4.25) circle (0cm) node[anchor=south]{$c_4^2$};
        \draw (25.3, -4.25) circle (0cm) node[anchor=south]{$c_5^1$};
        \draw (27.4, -4.25) circle (0cm) node[anchor=south]{$c_5^2$};
        \draw (29.3, -4.25) circle (0cm) node[anchor=south]{$c_6^1$};
        \draw (32.7, -4.25) circle (0cm) node[anchor=south]{$c_6^2$};

        \draw (5, -8.3) circle (0cm) node[anchor=south]{$d_1$};      
        \draw (12.7, -8.3) circle (0cm) node[anchor=south]{$d_2$};
        \draw (19.2, -8.3) circle (0cm) node[anchor=south]{$d_3$};
        \draw (23.2, -8.3) circle (0cm) node[anchor=south]{$d_4$};
        \draw (27.2, -8.3) circle (0cm) node[anchor=south]{$d_5$};
        \draw (31.2, -8.3) circle (0cm) node[anchor=south]{$d_6$};

        \draw (4, 8.5) circle (0cm) node[anchor=south]{$u_1$};
        \draw (32, 8.5) circle (0cm) node[anchor=south]{$u_2$};

        \draw (5, 13) circle (0cm) node[anchor=south]{$P_1$};
        \draw (31, 13) circle (0cm) node[anchor=south]{$P_2$};
        
        \draw (17.75, -10.5) circle (0cm) node[anchor=south]{$G$};

        \draw [decorate, decoration = {brace}]  (3.25, 12.5) -- (6.75, 12.5);
        \draw [decorate, decoration = {brace}]  (29.25, 12.5) -- (32.75, 12.5);

    \end{tikzpicture}
    \caption{Reduced instance of the SD problem constructed from the MRSS problem instance $S =\{(1,2), (2,1), (1,1)\}$, $t= (3,3), k = 2$ and $m = 2$. Here, $S'$ = $\{(2,3), (3,2), (2,2), (2,2), (2,2), (2,2)\}$ and $t' = (5,5)$.}
    \label{fig: Fig2}
\end{figure}
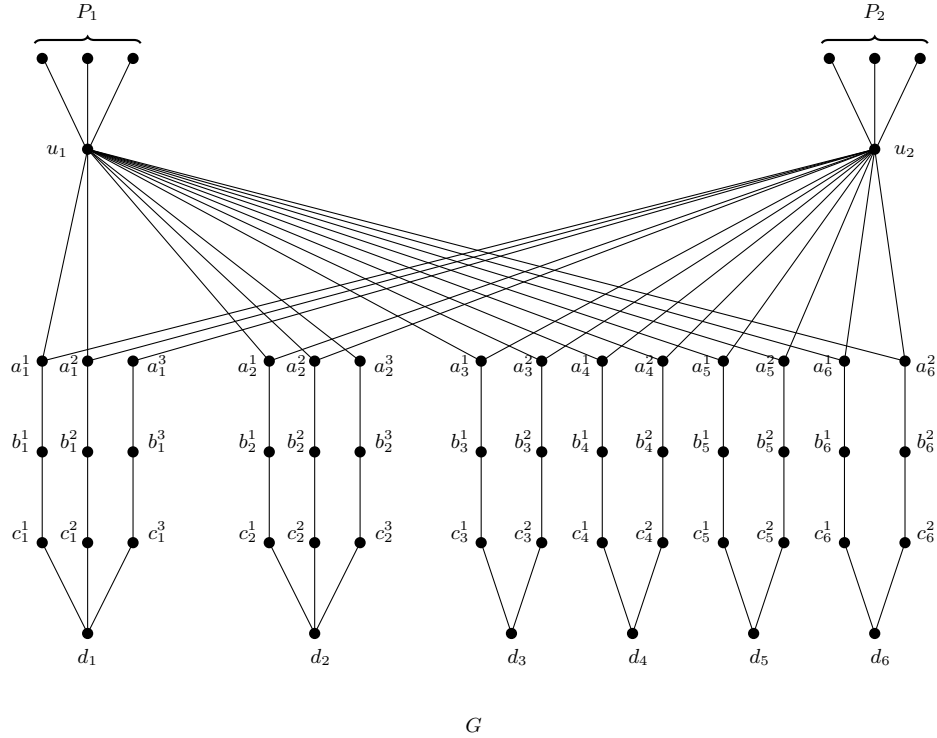
\begin{lemma}~\label{equal}
    ($k, m, S, t$) is a yes instance of the MRSS problem if and only if ($k, m, S', t'$) is a yes instance of the MRSS problem.
\end{lemma}
\begin{proof}
    By construction, the vectors in $S'$ obtained from $S$ have all coordinates at least $2$. Consequently, the additional $\max(t)$ copies of vector $(2,2)$ introduced in the construction of $S'$ do not affect any feasible solution. Moreover, since each coordinate of the vectors in $S$ is increased by one, the target vector $t$ must be adjusted accordingly by increasing each of its coordinates by $m$, yielding $t'$. Therefore, we conclude that ($k, m, S, t$) is a yes instance of the MRSS problem if and only if ($k, m, S', t'$) is a yes instance of the MRSS problem. \qed
\end{proof}

From Lemma \ref{equal}, it follows that ($k, m, S, t$) and ($k, m, S', t'$) are equivalent instances of the MRSS problem. Therefore, for the rest of this section, we consider the instance ($k, m, S', t'$) of the MRSS problem.
\begin{lemma}\label{sec4forward}
    \textit{If ($k, m, S', t'$) is a yes instance of }the MRSS problem\textit{ then there exists a signed dominating function $f$ of weight at most $k'$.}
\end{lemma}
\begin{proof}
     Let $R \subseteq S'$ such that $|R| \leq m$ and $\sum\limits_{s_i \in R}s_i \geq t'$. The corresponding signed dominating function $f$ is given as follows. \vspace{2mm} \\
     \[
f(w) = 
\begin{cases}
-1, & \text{if } w \in \bigcup\limits_{s_i \in R} B_{s_i} \cup \bigcup\limits_{s_i \notin R} (A_{s_i} \cup \{d_i\}), \\
1, & \text{if } w \in \bigcup\limits_{s_i \notin R} (B_{s_i} \cup C_{s_i}) \cup \bigcup\limits_{s_i \in R} (A_{s_i} \cup C_{s_i} \cup \{d_i\}) \cup \bigcup\limits_{j \in [k]} (P_j \cup \{u_j\}) \\
\end{cases}
\]
    \begin{itemize}
        \item For each $j \in [k]$, vertex $u_j$ has label 1 and each vertex in $P_j$ has label 1. This results in each vertex of $P_j$ having a labelSum of exactly two.
        \item For each $j \in [k]$, the negative weight from the neighbours of $u_j$ in $\bigcup\limits_{s_i \in S} A_{s_i}$ is at most the positive weight from $P_j$. As $u_j$ itself has label 1, the labelSum of $u_j$ remains positive.
        \item For $s_i \notin S$, and for each $l \in [\max(s_i)]$, the vertex $a_i^l$ has label $-1$, while $b_i^l$ and $c_i^l$ have label $1$. The vertex $d_i$ has label $-1$. Since all the neighbours of $a_i^l$, namely $b_i^l$ and the vertices in $\bigcup\limits_{j \in [k]} u_j$ have label $1$, the labelSum of $a_i^l$ is positive. The labelSum of $b_i^l$ is one, as its neighbours $a_i^l$ and $c_i^l$ have labels $-1$ and $1$, respectively. Similarly, the labelSum of $c_i^l$ is one, since its neighbours $b_i^l$ and $d_i$ have labels $1$ and $-1$, respectively. Finally, the labelSum of $d_i$ is positive, as all of its neighbours $\bigcup\limits_{l \in [\max(s_i)]} c_i^l$, have label 1.
        \item For $s_i \in S$, and for each $l \in [\max(s_i)]$, the vertices $a_i^l$ and $c_i^l$ have label $1$, while $b_i^l$ has label $-1$. The vertex $d_i$ has label $1$. Since all the neighbours of $a_i^l$ in $\bigcup\limits_{j \in [k]} u_j$ have label $1$, the labelSum of $a_i^l$ is positive. The labelSum of $b_i^l$ is one, as both of its neighbours $a_i^l$ and $c_i^l$ have label 1. Similarly, the labelSum of $c_i^l$ is one, since its neighbours $b_i^l$ and $d_i$ have labels $-1$ and $1$, respectively. Finally, the labelSum of $d_i$ is positive, as all of its neighbours $\bigcup\limits_{l \in [\max(s_i)]} c_i^l$, have label 1.
    \end{itemize}
The weight of $\bigcup\limits_{j \in [k]} u_j$ is $k$, that of $\bigcup\limits_{j \in [k]} P_j$ is $\sum\limits_{j  \in [k]} (\sum\limits_{s_i \in S'} s_i(j) - 2t'(j))$ and that of $\bigcup\limits_{i \in [n]}(A_{s_i} \cup B_{s_i} \cup C_{s_i} \cup \{d_i\})$ is $\max(s_i)-1$ if $s_i \notin R$, while it is $\max(s_i)+1$ if $s_i \in R$. Since each vertex in $G$ has a positive labelSum, we conclude that $f$ is a signed dominating function of weight at most $k' = \sum\limits_{j  \in [k]} (\sum\limits_{s_i \in S'} s_i(j) - 2t'(j))$+ $k +\sum\limits_{s_i \in S'}$ ($\max(s_i)-1$) + $2m$. \qed
\end{proof}
\begin{lemma} \label{lemma: L3}
    Let $f$ be a minimum weighed signed dominating function. For any $s_i \in S'$, the weight of $A_{s_i} \cup B_{s_i} \cup C_{s_i} \cup \{d_i\}$ is 
    \begin{enumerate}
    \item $\max(s_i)-1$ in $f$, if for each $l \in [\max(s_i)]$, $a_i^l$ is assigned label $-1$;
    \item $\max(s_i)+1$ in $f$, if for some $l \in [\max(s_i)]$, $a_i^l$ is assigned label $1$.
\end{enumerate}
\end{lemma}
\begin{proof}
    We distinguish two cases according to the labels of the vertices in $A_{s_i}$.
    \begin{enumerate}
        \item Suppose $f(a_i^l)=-1$ for all $l \in [\max(s_i)]\}$. For $b_i^l$ to have a positive labelSum, we must have $f(b_i^l) +f(c_i^l) \geq 2$ and hence $f(b_i^l) = f(c_i^l) = 1$. To obtain the minimum weight, the vertex $d_i$ is assigned label $-1$. This yields a weight of $\max(s_i)-1$ for $A_{s_i} \cup B_{s_i} \cup C_{s_i} \cup \{d_i\}$.
        \item Now suppose that \(f(a_i^l) = 1\) for some \(l \in [\max(s_i)]\). In this case, to ensure that the vertex \(c_i^l\) has a positive labelSum, at most one of the vertices \(b_i^l\), \(c_i^l\), and \(d_i\) may be labeled \(-1\), with the remaining two are assigned label \(1\). This increases the weight of \(A_{s_i} \cup B_{s_i} \cup C_{s_i} \cup \{d_i\}\) by at least two compared to the previous case. Therefore, the weight of \(A_{s_i} \cup B_{s_i} \cup C_{s_i} \cup \{d_i\}\) is at least \(\max(s_i)+1\).
    \end{enumerate} \qed
\end{proof}
\begin{lemma} \label{sec4backward}
    \textit{If there exists a signed dominating function $f$ of weight at most $k'$ then ($k, m, S', t'$) is a yes instance of }the MRSS problem.
\end{lemma}
\begin{proof}
    Let $f$ be a signed dominating function of weight $k' = \sum\limits_{j  \in [k]} (\sum\limits_{s_i \in S'} s_i(j) - 2t'(j))$+ $k +\sum\limits_{s_i \in S'}$ ($\max(s_i)-1$) + $2m$. 
    \begin{itemize}
    \item By Lemma \ref{obs: O1}, for each $j \in [k]$, we fix the label of $u_j$ as 1 and assign label $1$ to every vertex in $P_j$.
    \item After assigning labels to the vertices in $\bigcup\limits_{j \in [k]}(\{u_j\} \cup P_j)$, the vertex $u_j$ has labelSum $\sum\limits_{s_i \in S'} s_i(j) - 2t'(j)$. For $u_j$ to have a positive labelSum, we must assign label 1 to some of its neighbours in $\bigcup\limits_{s_i \in S'} A_{s_i}$. The labelSum of its neighbours in $\bigcup\limits_{s_i \in S'} A_{s_i}$ must be at least $2t'(j) - \sum\limits_{s_i \in S'} s_i(j)$.
    \item At this point, the remaining available weight for the set $\bigcup\limits_{s_i \in S'} (A_{s_i} \cup B_{s_i} \cup C_{s_i} \cup \{d_i\})$ is at most $\sum\limits_{s_i \in S'}$ ($\max(s_i)-1$) + $2m$. By Lemma \ref{lemma: L3}, for any $s_i \in S'$, if all the vertices in $A_{s_i}$ are labeled $-1$, then the minimum weight of $A_{s_i} \cup B_{s_i} \cup C_{s_i} \cup \{d_i\}$ is $\max(s_i)-1$. If instead, a subset of vertices in $A_{s_i}$ are labeled $1$, then the minimum weight increases to $\max(s_i)+1$, requiring an additional weight of two for each such set.
    \item Based on the remaining weight $\sum\limits_{s_i \in S'}$ ($\max(s_i)-1$) + $2m$ for $\bigcup\limits_{s_i \in S} (A_{s_i} \cup B_{s_i} \cup C_{s_i}\cup \{d_i\})$, we can afford to assign label $1$ to all vertices in $A_{s_i}$ for at most $m$ such sets.
    \item Hence, if there exist at most $m$ sets from $\bigcup\limits_{s_i \in S'}A_{s_i}$ such that assigning label $1$ to every vertex in each $A_{s_i}$ results in a positive labelSum for every vertex in $\bigcup\limits_{j \in [k]} u_j$, then there exist at most $m$ vectors in $\mathcal{S}$ whose sum is at least $t'$.
\end{itemize} \qed
\end{proof}

From Theorem \ref{mrss} and Lemmas \ref{sec4forward} and \ref{sec4backward}, we obtain the following theorem.
\begin{theorem}
    \sd{} problem is W[1]-hard parameterized by feedback vertex set number.
\end{theorem} 
The parameters treewidth and clique-width are bounded by a function of the parameter feedback vertex set number. With this direct relation between the parameters, we have the following theorem.
\begin{theorem} \label{thm: T2}
    \sd{} problem is W[1]-hard parameterized by treewidth or clique-width.
\end{theorem}

%% file: 5_nd.tex
\section{FPT algorithms for structural parameters}
In this section, we consider the structural parameters neighbourhood diversity and twin cover number, and present FPT algorithms.
\subsection{Neighbourhood diversity}
From Theorem \ref{thm: T2}, we have that the SD problem is W[1]-hard parameterized by clique-width. Here, we investigate its parameterized complexity with respect to more restrictive parameter, neighbourhood diversity, and present an FPT algorithm.

\medskip
For a graph $G = (V, E)$, \cite{lampis} showed that the neighbourhood diversity can be computed in polynomial-time. For the rest of this section, we assume that a partition of $V$ into $t$ sets $V_1, V_2, ..., V_t$ is given to us. By the definition of neighbourhood diversity, each set $V_i$ in the partition is either a clique or an independent set. 

\medskip
To establish fixed-parameter tractability, we transform the problem into an instance of Integer Linear Programming (ILP) problem, which is known to be FPT in the number of variables. The existing result related to the ILP problem is given as follows.
\begin{theorem} [\cite{FLWS}] ~\label{theoremilp}
    The $p$-variable Integer Linear Programming problem can be solved using $\mathcal{O}(p^{2.5p+o(p)} \cdot L \cdot \log(MN))$ arithmetic operations and space polynomial in $L$, where $L$ is the number of bits in the input, $N$ is the maximum absolute value any variable can take, and $M$ is an upper bound on the absolute value of the minimum taken by the objective function.
\end{theorem} 
 We now present the details for constructing an ILP formulation for the SD problem parameterized by neighborhood diversity.
\begin{itemize}
    \item Let $V_1, V_2, ..., V_t$ denote a partition of the vertex set $V$ into corresponding sets. We use $C$ to denote the union of cliques and $I$ to denote the union of independent sets among $V_1, V_2, ..., V_t$.
    \item Let $f$ be a minimum weighted signed dominating function. For each partition $V_i$, we define two boolean variables $a_i$ and $b_i$, where $a_i = 1$ indicates that there exists a vertex in $V_i$ with label $-1$ under $f$ and $b_i = 1$ indicates that there exists a vertex in $V_i$ with label $1$ under $f$.
    \item For each partition $V_i$, we guess whether there exists a vertex with label $-1$. If such a vertex exists, we set $a_i = 1$. Similarly, we guess whether there exists a vertex with label 1. If such a vertex exists, we set $b_i = 1$. If $a_i= 0$ and $b_i =0$ for any partition, we reject the guess.
    \item For each partition $V_i$, we define a variable $x_i$ representing the sum of the labels of all vertices in the partition. For some partitions, the values of $x_i$ do not form a continuous range. To model the corresponding lower and upper bounds as a continuous range in the ILP formulation, we introduce an auxiliary variable $y_i$ for such partitions.
\end{itemize}

    \medskip
    
\noindent The ILP formulation for the SD problem is given as follows.
\begin{tcolorbox}
{
    Minimize 
    $\sum\limits_{i: a_i= 1, b_i=1, |V_i| \equiv 0 \pmod 2} 2y_i$ +
    $\sum\limits_{i: a_i= 1, b_i=1, |V_i| \equiv 1 \pmod 2}$ $(2y_i+1)$
    
    \noindent Subject to
    \begin{enumerate}
    \item $x_i + A \geq 1$, for each $V_i \in C$,
    \item $-1 + A \geq 1$, for each $V_i \in I$ with $a_i = 1$ and $b_i = 0$,
    \item $1 + A \geq 1$, for each $V_i \in I$ with $a_i = 0$ and $b_i= 1$,
    \item $\left \lfloor \frac{-|V_i|}{2} \right \rfloor +1\leq y_i \leq \left \lfloor \frac{|V_i|}{2} \right \rfloor -1$, for each $V_i$ with $a_i = 1$ and $b_i = 1$,
    \end{enumerate}
    Where $A$ =  $-\sum\limits_{j: V_j \in N(V_i), a_j = 1, b_j = 0} |V_j|$ + 
    $\sum\limits_{j: V_j \in N(V_i), a_j = 0, b_j=1} |V_j|$ + 
    $\sum\limits_{j: V_j \in N(V_i), a_j= 1, b_j=1, |V_j| \equiv 0 \pmod 2} 2y_j$ +
    $\sum\limits_{j: V_j \in N(V_i), a_j= 1, b_j=1, |V_j| \equiv 1 \pmod 2}$ $(2y_j+1)$
}
\end{tcolorbox} 
\noindent We now discuss the constraints given in the ILP formulation. 
    \medskip
    
\noindent 1. For each partition $V_i$ that is a clique, the sum of the labels of vertices in $V_i$ together with the sum of the labels of vertices in its neighbouring partitions must be at least one. 
    \medskip
    
\noindent 2. For each partition $V_i$ that is an independent set, if $a_i = 1$ and $b_i$ = 0, then there exists a vertex in $V_i$ with label $-1$ and no vertex with label $1$.
To ensure a positive labelSum for all the vertices in $V_i$, $-1$ plus the sum of the labels of the vertices in its neighboring partitions must be at least one. 
    \medskip
    
\noindent 3. For each partition $V_i$ that is an independent set, if $a_i = 0$ and $b_i = 1$, then there exists a vertex in $V_i$ with label 1 and no vertex with label $-1$. To guarantee a positive labelSum for all the vertices in $V_i$, the value $1$ plus the sum of the labels of the vertices in its neighbouring partitions must be at least one. 
    \medskip
    
\noindent 4. For each partition $V_i$ with $a_i = 1$ and $b_i = 1$, the value of $x_i$ is minimum if there are $|V_i|-1$ vertices with label $-1$ and one vertex with label $1$. Therefore, the lower bound for $x_i$ is ${-|V_i|} +2$. The value of $x_i$ is maximum if there are $|V_i|-1$ vertices with label $1$ and one vertex with label $-1$. Hence, the upper bound for $x_i$ is ${|V_i|} -2$. 

If $|V_i|$ is even then $x_i$ can only take even values between ${-|V_i|} +2$ and ${|V_i|} -2$. Thus, we write $x_i$ = $2y_i$, where $y_i$ $ \in \left[\left \lfloor \frac{-|V_i|}{2} \right \rfloor +1, \left \lfloor \frac{|V_i|}{2} \right \rfloor -1 \right]$.  Whereas if $|V_i|$ is odd then $x_i$ can only take odd values between ${-|V_i|} +2$ and ${|V_i|} -2$. Thus, we write $x_i$ = $2y_i+1$, where $y_i$ $ \in \left[\left \lfloor \frac{-|V_i|}{2} \right \rfloor +1, \left \lfloor \frac{|V_i|}{2} \right \rfloor -1 \right]$.
    \medskip
    
Next, we describe how labels are assigned to the vertices in each partition $V_i$, based on the guesses for $a_i$ and $b_i$. 
\begin{itemize}
    \item If $a_i = 0$ and $b_i = 1$, we assign label 1 to every vertex in $V_i$. 
    \item If $a_i = 1$ and $b_i = 0$, we assign label $-1$ to every vertex in $V_i$. 
\end{itemize}
    For all other partitions, labels are assigned based on the outcome of the ILP formulation. 
\begin{itemize}
    \item If $a_i = 1$ and $b_i = 1$, we assign labels $\{-1, 1\}$ to the vertices in $V_i$ such that the weight of $V_i$ is exactly $x_i$ and there must exist two vertices $u$ and $v$ in $V_i$ such that $f(u) = -1$ and $f(v) = 1$.
\end{itemize}
In this way, we label all vertices of the partitions $V_1, V_2, ..., V_t$.

\medskip
In our ILP formulation, we have $t$ variables. The values of all the variables and the objective function are upper bounded by $n$. By Theorem \ref{theoremilp}, the problem can thus be solved in $\mathcal{O}^*(4^t \cdot t^{\mathcal{O}(t)})$ time.

\medskip
Hence, we arrive at the following theorem. 
\begin{theorem}
    SD problem is FPT parameterized by neighbourhood diversity.
\end{theorem}

%% file: 6_tc.tex
\subsection{Twin cover number}
We investigate the parameterized complexity with respect to the structural parameter twin cover number, and present an FPT algorithm.

\medskip
The parameter twin cover of size at most $k$ (if exists), can be obtained in FPT time, in the following way.
\begin{theorem}[{{\hspace{-1sp}\cite[Theorem 4]{GNN}}}] \label{thrm: T5_2} \textit{If a minimum twin cover in $G$ has size at most $k$, then it is possible to compute a twin cover of size at most $k$ in time $\mathcal{O}(|E|+k|V|+1.2738^k)$.} \end{theorem}
Given a graph $G$, we partition the vertex set $V$ into sets $T$ and $C$, where $T$ represents the twin cover and $C$ is the union of clique sets outside $T$. A clique set is a union of cliques with the same adjacency in $T$. From here on, we use $t$ in place of $|T|$.

\medskip
Let $f$ be a minimum weighted signed dominating function and $P$ be the set of vertices in $T$ with the label 1 under $f$. We guess the set $P$ in $2^{t}$ ways. For a clique set $C_i$, let $y_i$ denote the sum of the labels of the vertices in $N(C_i) \cap T$. Let $K$ be a clique in $C_i$. If $|K|+y_i<1$, then no feasible solution exists, and we reject such guess.

\medskip
For each vertex $u \in T$, we define $z(u)$ to be the sum of the labels of its closed neighbours in $T$ and we use $M(u)$ to denote the indices of the clique sets $u$ is adjacent to. 
\medskip
\begin{lemma}\label{parity} Let $C_i$ be a clique set and $y_i$ be the sum of the labels assigned to the vertices in $N(C_i)\cap T$. let $K$ be a clique in $C_i$.
\begin{enumerate}
    \item If $y_i$ is odd, then
    \begin{itemize}
   \item when $|K|$ is odd, the weight of $K$ is an odd integer and satisfies
     $
     w(K)\ge 2-y_i,
     $
   \item when $|K|$ is even, the weight of $K$ is an even integer and satisfies
     $
     w(K)\ge 1-y_i.
     $
    \end{itemize}
 \item If $y_i$ is even, then
    \begin{itemize}
   \item when $|K|$ is odd, the weight of $K$ is an odd integer and satisfies
     $
     w(K)\ge 1-y_i,
     $
   \item when $|K|$ is even, the weight of $K$ is an even integer and satisfies
    $
     w(K)\ge 2-y_i.
     $
     
    \end{itemize}
\end{enumerate}
     \end{lemma}

\begin{proof}
Let $K$ be a clique in the clique set $C_i$. Since all the vertices of $K$ have the same neighborhood in $T$, every vertex $v\in K$ receives the same contribution $y_i$ from vertices outside $K$.

Let $w(K)$ denote the sum of the labels assigned to the vertices of $K$. We have, $w(K)+y_i \ge 1$. Hence, $w(K)\ge 1-y_i.$

\medskip
Next, let us assume that $a$ vertices of $K$ has label $1$ and $b$ vertices has label $-1$ under $f$, then
$w(K)=a-b$ and $|K|=a+b.$ Therefore, $w(K)\equiv |K| \pmod 2,$ so $w(K)$ has the same parity as $|K|$.

\medskip
Consequently, the minimum feasible value of $w(K)$ is the smallest
integer having the same parity as $|K|$ and at least $1-y_i$.

\begin{enumerate}
    \item If $y_i$ is odd, then $1-y_i$ is even.
    \begin{itemize}
        \item If $|K|$ is odd, the smallest odd integer at least
        $1-y_i$ is $2-y_i$.
        \item If $|K|$ is even, the smallest even integer at least
        $1-y_i$ is $1-y_i$.
    \end{itemize}

    \item If $y_i$ is even, then $1-y_i$ is odd.
    \begin{itemize}
        \item If $|K|$ is odd, the smallest odd integer at least
        $1-y_i$ is $1-y_i$.
        \item If $|K|$ is even, the smallest even integer at least
        $1-y_i$ is $2-y_i$.
    \end{itemize}
\end{enumerate} \qed

\end{proof}



\medskip

\medskip
\noindent
For every clique $K\in C_i$, let $\ell(K)$ denote the minimum feasible
weight of $K$ obtained from Lemma~\ref{parity}. Then feasible weight
for $K$ can be written as $
w(K)=\ell(K)+2q_K,$
where $q_K\geq 0$.

\medskip
For each clique set $C_i$, we define
$
L_i=\sum\limits_{K\in C_i}\ell(K).
$
Here, $L_i$ can be computed, once the guess of $P$ is fixed. 
Let
$
x_i=\sum\limits_{K\in C_i} q_K.
$
Then the total contribution of the clique set $C_i$ is
$
L_i+2x_i.
$

\medskip
The ILP formulation for the SD problem is given as follows.
\begin{tcolorbox}

Minimize 
$\sum\limits_{i=1}^{t} x_i$
\vspace{2mm}

Subject to
\begin{enumerate}
    \item $z(u)+
\sum\limits_{i\in M(u)}
\bigl(L_i+2x_i\bigr)
\ge 1,$ for each $u\in T$,
\item $0\le x_i\le m_i,$ for each $i\in [t]$,
\end{enumerate}
Where 
$
m_i=
\sum\limits_{K\in C_i}
\frac{|K|-\ell(K)}{2}.
$
\end{tcolorbox}

\noindent We now discuss the constraints given in the ILP formulation. 
    \medskip
    
\noindent 1. For every vertex $u \in T$, the sum of the labels of its closed neighbours in $T$ together with the sum of the labels of the vertices in the clique sets adjacent to $u$, must be at least one. Hence, we have that $z(u)+
\sum\limits_{i\in M(u)} \bigl(L_i+2x_i\bigr) \ge 1.$
    \medskip
    
\noindent 2. For each clique set $C_i$, the variable $x_i$ takes the values in the range $0$ and $m_i$, where $m_i$ is defined as half, over all the cliques $K \in C_i$, of the difference between the size of $K$ and $l(K)$.
    
    \medskip
\noindent Now we describe, how to assign labels to the vertices in $C_i$ under $f$, from the outcome of the ILP formulation.

\medskip
For every clique $K \in C_i$, we assign label $1$ to $l(K)$ vertices arbitrarily. We then label additional $2x_i$ vertices in $C_i$ with $1$ arbitrarily. All the other vertices in $C_i$ will be labeled $-1$.

\medskip
In our ILP formulation, we have $2^t$ variables. The values of all the variables and the objective function are upper bounded by $n$. By Theorem \ref{theoremilp}, the problem can thus be solved in $\mathcal{O}^*(2^t \cdot 2^{{\mathcal{O}(t \cdot 2^t)}})$ time.

\medskip
Hence, we obtain the following result. 
\begin{theorem}
    SD problem is FPT parameterized by twin cover number.
\end{theorem}

%% file: 7_conclusion.tex
\section{Conclusion}
In this paper, we studied the complexity aspects of the SD problem. We proved that the problem remains NP-complete even when restricted to split graphs. We then showed that the problem is W[1]-hard parameterized by feedback vertex set number. Finally, we presented FPT algorithms for the parameters neighbourhood diversity and twin cover number.
\medskip

\noindent We conclude the paper with the following open questions on the SD problem.
\begin{enumerate}
    \item What is the complexity of the problem on interval graphs, strongly chordal graphs, and block graphs?
    \item What is the parameterized complexity of the problem parameterized by distance to disjoint paths, distance to cluster, pathwidth, and modular-width?
    \item Does there exist polynomial kernels for the problem parameterized by vertex cover number, distance to clique, and max-leaf number?
\end{enumerate}